\documentclass{revtex4}
\usepackage{amsmath,amssymb,amsfonts}
\usepackage{algorithm}
\usepackage{tikz,pgf}
\usepackage{multirow}
\usepackage{epstopdf}
\usepackage{url,verbatim}

\newtheorem{definition}{Definition}
\newtheorem{theorem}{Theorem}
\newtheorem{proof}{Proof}
\newtheorem{lemma}{Lemma}

\begin{document}
\newcommand{\ket}[1]{\ensuremath{\left|#1\right\rangle}}
\newcommand{\bra}[1]{\ensuremath{\left\langle#1\right|}}
\newcommand\floor[1]{\lfloor#1\rfloor}
\newcommand\ceil[1]{\lceil#1\rceil}
\newcommand{\tr}{\mathrm{Tr}}
\title{Quantum secure two party computation for set intersection with rational players}
\author{Arpita Maitra}
\affiliation{Centre for Theoretical Studies, Indian Institute of Technology Kharagpur, \\
Kharagpur 721302, West Bengal, India.\\
	\email{arpita76b@gmail.com}}

\date{Received: date / Accepted: date}

\begin{abstract}
Recently, Shi et al. (Phys. Rev. A, 2015) proposed Quantum Oblivious Set Member Decision Protocol (QOSMDP) where two legitimate parties, namely Alice and Bob, play a game. Alice has a secret $k$ and Bob has a set $\{k_1,k_2,\cdots k_n\}$. The game is designed towards testing if the secret $k$ is a member of the set possessed by Bob without revealing the identity of $k$. The output of the game will be either ``Yes'' (bit $1$) or ``No'' (bit $0$) and is generated at Bob's place. Bob does not know the identity of $k$ and Alice does not know any element of the set. In a subsequent work (Quant. Inf. Process., 2016), the authors proposed a quantum scheme for Private Set Intersection (PSI) where the client (Alice) gets the intersected elements with the help of a server (Bob) and the server knows nothing. In the present draft, we extended the game to compute the intersection of two computationally indistinguishable sets $X$ and $Y$ possessed by Alice and Bob respectively. We consider Alice and Bob as rational players, i.e., they are neither ``good'' nor ``bad''. They participate in the game towards maximizing their utilities. We prove that in this rational setting, the strategy profile $((cooperate, abort), (cooperate, abort)$) is a strict Nash equilibrium. If $((cooperate, abort), (cooperate, abort)$) is strict Nash, then fairness as well as correctness of the protocol are guaranteed.
\end{abstract}
\maketitle
\noindent{\bf Keywords:} Set Intersection;  Quantum Secure Computation; Rationality; Nash Equilibrium 

\section{Introduction}
\label{intro}
Secure Multiparty Computation (SMC)~\cite{Gordon,AL11,GroceK}  is an important primitive in cryptology. It  has wide applications in electronic voting, cloud computing, online auction etc. Recent trend of the theoretical research in this direction is to combine game theory with cryptology. 

Cryptography deals with `worst case' scenario making the protocols secure against various type of adversarial behaviours. Those are characterized as semi honest, malicious and covert adversarial models.  

In game theory, a protocol is designed against the rational deviation of a party. Rational parties are neither completely ``malicious'' nor they are fully ``honest''. They participate in the game in the motivation towards maximizing their utilities. So placing cryptographic protocols in rational setting empowers more flexibility to the adversary. It seems more practical as in real world most of the people prefer to be rational rather being completely ``good'' or ``bad''.    

 Recently, Brunner and Linden~\cite{BL} showed a deep link between quantum physics and game theory.  They showed that if the players use quantum resources, such as entangled quantum particles, they can outperform the classical players. In~\cite{maitra} the concept of rationality has been introduced in quantum secret sharing. In the present draft, we incorporate the idea of rationality in Secure Computation for Set Intersection (SCSI). 
 
 In classical domain this SCSI problem has been studied extensively~\cite{HZ1,HZ2,JN10}. It has various applications in dating services, data-mining, recommendation systems, law enforcement etc. 
 
 In SCSI, two parties, Alice and Bob, hold two sets $X$ and $Y$ respectively. The sets are assumed to be computationally indistinguishable from each other. Alice and Bob exchange some informations between themselves so that at the end of the protocol, either Alice or Bob (suggested by the protocol) gets $X \cap Y$. 
 
 However, the hardness assumptions that Diffie--Hellman (DDH) problem~\cite{Diffi-Hellman}, discret logarithm (DL) problem~\cite{Stingson} are computationally hard, have been proven to be vulnerable in quantum domain~\cite{Shor}.
 
 In quantum paradigm, Shi et al.~\cite{Shi} proposed a variant of this problem and named it as {\em {Quantum Oblivious Set Member Decision Protocol}} (QOSMDP). According to the protocol, Alice has a secret element $k$ and Bob holds a set $\{k_1,k_2.\cdots,k_n\}$ of $n$ elements. Now, Bob wants to know if the secret $k$ of Alice is the member of his set. However, Alice does not allow Bob to identify that element. Simultaneously, Bob resists Alice to know a single element except $k$, if it is in the set, of the set. 
 
 The authors of~\cite{Shi} commented that it can be exploited to compute the cardinality of the set intersection or union which is the direct consequence of the protocol. Even in~\cite{Shi1}, the authors suggested a quantum scheme for Private Set Intersection (PSI) where the client (Alice) gets $X\cap Y$ with the help of a server (Bob) and the server knows nothing. They establish the security of their protocol in ``honest but curious'' adversarial model. 
 
 Contrary to this, in the present draft, we exploit the idea to compute the set intersection in rational setting where the players are trying to maximize their utilities. We show that the strategy profile $((cooperate, abort), (cooperate, abort))$ achieves strict Nash equilibrium in this initiative. We also show that if ((cooperate, abort), (cooperate, abort)) is a strict Nash, then fairness as well as correctness of the protocol are obvious. 

In this regard, we like to point out that the procedure described in~\cite{Shi,Shi1} to detect {\em{measure and resend}} attack by Bob, requires a more detailed analysis and certain revision. In~\cite{Shi}, to detect the attack, the author inserted $l-1$ decoy states encoded as $\frac{1}{\sqrt{2}}(\ket{0}+\ket{j_i})$, where $j_i\in Z_{N}^*$. Each $j_i$ is represented by $\log_2 N$ bits. In~\cite{Shi1}, the same type of encoding is exploited to mask the set elements of the client's set. In both the papers, it is commented that if Bob tries to measure those states, he will introduce noise in the channel which can be detected by Alice by measuring the states in $\{\frac{1}{\sqrt{2}}(\ket{0}+\ket{j_i}),\frac{1}{\sqrt{2}}(\ket{0}-\ket{j_i})\}$ basis. However, the oracle $O_t$ maps the state $\frac{1}{\sqrt{2}}(\ket{0}+\ket{j_i})$ into $\frac{1}{\sqrt{2}}(\ket{0}+(-1)^{q_t(j_i)}\ket{j_i})$. Based on the value of $q_t(j_i)\in \{0,1\}$ (not known to Alice), the state will be either $\frac{1}{\sqrt{2}}(\ket{0}+\ket{j_i})$ or $\frac{1}{\sqrt{2}}(\ket{0}-\ket{j_i})$. Thus, it requires further clarification how Alice can distinguish the attack from the application of the oracle by measuring the registers in $\{\frac{1}{\sqrt{2}}(\ket{0}+\ket{j_i}), \frac{1}{\sqrt{2}}(\ket{0}-\ket{j_i})\}$ basis. To avoid such a security related issue, we modify the protocol accordingly. 

\section{Preliminaries}
In this section, we discuss the concepts of computational indistinguishability of two distribution ensembles, functionality, rationality, fairness, correctness and equilibrium used in this work. 

\subsection{Computational Indistinguishability} In communication complexity, two distribution ensembles $X=\{X(a,\lambda)\}_{a\in\{0,1\}^*}$ and $Y=\{Y(a,\lambda)\}_{a\in\{0,1\}^*}$ (where $\lambda$ is the security parameter which usually refers to the length of the input), are computationally indistinguishable if for any non-uniform probabilistic polynomial time algorithm $D$, the following quantity is a negligible function in $\lambda$:

$$\delta (\lambda)=\left|\Pr _{a\gets X(a,\lambda)}[D(a)=1]-\Pr _{a\gets Y(a,\lambda)}[D(a)=1]\right|$$ for every $\lambda\in N$.

In other words, two ensembles are computationally indistinguishable implies that those can not be distinguished by polynomial-time algorithms looking at multiple samples taken from those ensembles. 
\subsection{Functionality}
In classical domain and in two party setting, a functionality $\mathcal{F}=\{f_\lambda\}_{\lambda\in\mathbb{N}}$ is a sequence of randomized processes, where $\lambda$ is the 
security parameter and $f_\lambda$ maps pairs of inputs to pairs of outputs (one for each party). Explicitly, we can write 
$f_\lambda=(f_\lambda^{1}, f_\lambda^{2})$, where $f_\lambda^{1}$ represents the output of the first party, say Alice. Similarly, $f_\lambda^{2}$ represents the output of the second 
party, say Bob. The domain of $f_\lambda$ is $X_\lambda \times Y_\lambda$, where $X_\lambda$ (resp. $Y_\lambda$) denotes the possible inputs of the first (resp. second) party. 
If $|X_\lambda|$ and $|Y_\lambda|$ are polynomial in $\lambda$, then we say that $\mathcal{F}$ is defined over polynomial size domains. If each $f_\lambda$ is deterministic, we say 
that each $f_\lambda$ as well as the collection $\mathcal{F}$ is a function~\cite{GK11}.

\subsection{Rationality}
\label{rat}
Rationality of a player is defined over its utility function $U\in \{u_1,u_2,\cdots,u_n\}$ and its preferences. Each $u_i$, $i \in \{0,1,\cdots,n\}$ is associated with the possible outcomes of the game. The outcomes and corresponding utilities for $2$ players' set intersection game are described in Table~\ref{table: OutcomesRSS}. 

$f_A$ (resp. $f_B$) represents the functionality generated at the place of Alice (resp. Bob) and $U_A$ (resp. $U_B$) represents the utility function of Alice (resp. Bob). Let $\mathcal{F}=X\cap Y$ and $\perp=\emptyset$

\begin{table}[htbp]
\caption{}
\label{table: OutcomesRSS}
\begin{center}
\begin{tabular}{llll}
\hline\noalign{\smallskip}
$f_A$ & $f_B$ & $U_A(f_A, f_B)$ & $U_B(f_A, f_B)$\\
\hline
\noalign{\smallskip}
$f_A=\mathcal{F}$ & $f_B$=$\mathcal{F}$ & $U_A^{TT}$ & $U_B^{TT}$\\ \\

$f_A=\perp$ & $f_B=\perp$ & $U_A^{NN}$ & $U_B^{NN}$\\ \\

$f_A=\mathcal{F}$ & $f_B=\perp$ & $U_A^{TN}$ & $U_B^{NT}$\\ \\

$f_A=\perp$ & $f_B=\mathcal{F}$ & $U_A^{NT}$ & $U_B^{TN}$\\ \\
\hline
\end{tabular}
\end{center} 
\end{table}

Here, $TT$, $TN$, $NT$, $NN$ imply 
\begin{itemize}
\item both Alice and Bob obtain ``True'' output, i.e., correct and complete values of $\mathcal{F}$. 
\item Alice (resp. Bob) obtains ``True'' output where Bob (resp. Alice) obtains ``Null'' output, 
\item Alice (resp. Bob) obtains ``Null'' output but Bob (resp. Alice) gets ``True'' output, 
 \item both obtain ``Null'' output
\end{itemize}
respectively.
 In this work, we assume that Alice (resp. Bob) has the following order of preferences.
$$\mathcal{R}_1 : U^{TN}> U^{TT}>U^{NN}>U^{NT}.$$ 
That is they prefer to compute the true value of the functionality by herself or himself alone than to compute the true value by both. However, they find it better to compute a null value at both of their ends than to compute a null value by himself or herself when the opponent gets a true value.

Here, one should emphasize that each rational party is only interested to get the complete value of functionality $\mathcal{F}$.

\subsection{Fairness}
A rational player, being selfish, desires an unfair outcome, i.e., he or she always tries to compute the true value of the functionality by himself or herself alone. Therefore, the basic aim of a game when the players are rational should be to achieve fairness. 

A formal definition of fairness in the context of a
(2,2) Rational Secret Sharing (RSS) protocol was presented by Asharov and Lindell~\cite{AL}. We modify this definition accordingly for our present setting.
\begin{definition} Let $\sigma$ be the strategy suggested by the protocol and $\sigma'$ be any deviated strategy.
Suppose, Alice has a strategy profile $(\sigma_A,\sigma'_A)$. Similarly Bob has a strategy profile $(\sigma_B,\sigma'_B)$. A game is said to be completely fair if for every arbitrary alternative strategy $\sigma'_A$ followed by Alice, the following holds:
\begin{eqnarray*}
\Pr[f_A=\mathcal{F}|A=\sigma'_A, B=\sigma_B] \\
< \Pr[f_A=\mathcal{F}|A=\sigma_A, B=\sigma_B].
\end{eqnarray*}
\end{definition}
Here $A$ (resp. $B$) implies the event that Alice (resp. Bob) follows a strategy. 

Similarly, for Bob we can write
\begin{eqnarray*}
\Pr[f_B=\mathcal{F}|A=\sigma_A, B=\sigma'_B] \\
< \Pr[f_B=\mathcal{F}|A=\sigma_A, B=\sigma_B].
\end{eqnarray*}

In terms of utility function, a game achieves fairness if and only if for a party, the following holds:
\begin{eqnarray*}
U^{TT} \geq E[U (\mathcal{O}_i)], 
\end{eqnarray*}
where, E(U) is the expected utility value of the player for the input $i$, $i\in\{1,\cdots,n\}$ and $\mathcal{O}_i$ is the corresponding outcome.

\subsection{Correctness}\label{correctness}
A formal definition of correctness in the context of a
(2,2) RSS protocol was presented by Asharov and Lindell~\cite{AL}. We modify this definition for the setting as follows:
\begin{definition}
(Correctness): Let $\sigma$ be the strategy suggested by the protocol and $\sigma'$ be any deviated strategy. Let Alice has a strategy profile $(\sigma_A,\sigma'_A)$. Similarly Bob has a strategy profile $(\sigma_B,\sigma'_B)$. A game is said to be correct if for every arbitrary alternative strategy $\sigma'_B$ followed by Bob, the following holds:
\begin{eqnarray*}
\Pr[f_A\not\in \{\mathcal{F},\perp\}|A=\sigma_A, B=\sigma'_B]= 0
\end{eqnarray*}
\end{definition}
Here $A$ (resp. $B$) implies the event that Alice (resp. Bob) follows a strategy. 

Similarly, for Bob we can write
\begin{eqnarray*}
\Pr[f_B\not\in \{\mathcal{F},\perp\}|A=\sigma'_A, B=\sigma_B]=0
\end{eqnarray*}

\subsection{Equilibrium}
  Let $\Gamma$ be a mechanism designed for $n$ players for a certain purpose. Let $\overrightarrow{\sigma}$ be the set of suggested strategies for that $n$ number of players in the mechanism $\Gamma$. That is $\overrightarrow{\sigma}=\{\sigma_1,\sigma_2,\cdots,\sigma_n\}$, where $\sigma_i$ is the suggested strategy for a player $P_i$, $i\in\{1,2,\cdots,n\}$. Then  $\overrightarrow{\sigma}$ in the mechanism $(\Gamma,\overrightarrow\sigma)$ is said to be in Nash equilibrium when there is no incentive for a player $P_i$, $i\in\{1,2,\cdots,n\}$ to deviate from the suggested strategy, given that everyone else is following his or her strategy. Thus we can define 
Strict Nash Equilibrium as follows.
\begin{definition}
(Strict Nash Equilibrium) The suggested strategy $\overrightarrow{\sigma}$ in the mechanism $(\Gamma,\overrightarrow{\sigma})$ is a strict Nash equilibrium if for every $P_i$ and for any strategy $\sigma'_i$, we have $u_i (\sigma'_i,\overrightarrow{\sigma}_{-i} )<u_i (\overrightarrow{\sigma})$.
\end{definition}
Here, $\overrightarrow{\sigma}_{-i}=\{\sigma_1,\sigma_2,\cdots,\sigma_{i-1},\sigma_{i+1},\sigma_n\}$, i.e., the set of the suggested strategies for the players excluding $i$-th player.

Explicitly, a mechanism is in strict Nash equilibrium when the payoff achieved by a player following the suggested strategy will be more than the payoff achieved by the player following any deviated strategy conditional on the event that all other players follow the suggested strategies.
\section{Revisiting the protocol in~\cite{Shi}}
In~\cite{Shi} the protocol for set member decision problem is described as follows.
Alice has a secret $k$ and Bob possesses a set $Y=\{k_1,k_2,\cdots,k_n\}$ such that each $k_i$ belongs to the set $\mathbb{Z}_{N}^*=\{1,2,\cdots,N-1\}$. Now, Bob prepares an $N$ element database in a way so that the $j$-th element $p(j)=1$ if and only if $j=k_i(i\in[1,n])$ and $p(j)=0$ otherwise. He now selects $l$ bits $r_1,r_2,\cdots,r_l\in\{0,1\}$ uniformly at random and generates another variable $q_t(j)=p(j)\oplus r_t$ where $t$ varies from $1$ to $l$ and $j$ varies from $1$ to $N-1$. $l$ is the security parameter. Alice and Bob fix $p(0)=0$ and $q_1(0)=q_2(0)=\cdots=q_l(0)=0$ a priori. Table $2$ shows the $N-1$ elements database created from the set $\{k_1,k_2,\cdots,k_n\}$.

\begin{table*}[htbp]
\begin{center}
\begin{tabular}{|cccccc|}
\hline
 $j$ & $p(j)$ & $q_1(j)$ & $q_2(j)$ & $\cdots$ & $q_l(j)$\\
 \hline
1 & 0 & $0+r_1$ & $0+r_2$ & $\cdots$ & $0+r_l$\\
2 & 0 & $0+r_1$ & $0+r_2$ & $\cdots$ & $0+r_l$\\
3 & 0 & $0+r_1$ & $0+r_2$ & $\cdots$ & $0+r_l$\\
$\vdots$ & $\vdots$ & $\vdots$ & $\vdots$ & $\vdots$ & $\vdots$\\
$k_1$ & 1 & $1+r_1$ & $1+r_2$ & $\cdots$ & $1+r_l$\\
$k_1+1$ & $0$ & $0+r_1$ & $0+r_2$ & $\cdots$ & $0+r_l$\\
$\vdots$ & $\vdots$ & $\vdots$ & $\vdots$ & $\vdots$ & $\vdots$\\
$k_2$ & 1 & $1+r_1$ & $1+r_2$ & $\cdots$ & $1+r_l$\\
$k_2+1$ & $0$ & $0+r_1$ & $0+r_2$ & $\cdots$ & $0+r_l$\\
$\vdots$ & $\vdots$ & $\vdots$ & $\vdots$ & $\vdots$ & $\vdots$\\
$k_n$ & 1 & $1+r_1$ & $1+r_2$ & $\cdots$ & $1+r_l$\\
$k_n+1$ & $0$ & $0+r_1$ & $0+r_2$ & $\cdots$ & $0+r_l$\\
$\vdots$ & $\vdots$ & $\vdots$ & $\vdots$ & $\vdots$ & $\vdots$\\
$N-1$ & 0 & $0+r_1$ & $0+r_2$ & $\cdots$ & $0+r_l$\\
\hline
\end{tabular}
\end{center}
\caption{$N$ element database created from the set $\{k_1,k_2,\cdots,k_n\}$~\cite{Shi}}
\label{tab2}
\end{table*}

Alice now generates $l$ $M(=\log_2 N)$\footnote{For the brevity of notation, in the rest of the paper, we will write $\log(.)$ instead of $\log_2(.)$} qubit registers. One register contains the qubit $\frac{1}{\sqrt{2}}(\ket{0}+\ket{k})$ and the remaining $l-1$ registers contain the decoy states $\frac{1}{\sqrt{2}}(\ket{0}+\ket{j_1}), \frac{1}{\sqrt{2}}(\ket{0}+\ket{j_2}),\cdots,\frac{1}{\sqrt{2}}(\ket{0}+\ket{j_{l-1}})$ where $j_i\in \mathbb{Z}_N^*$. Here, $\ket{0}$ represents $\ket{0}^{\otimes M}$ and each $\ket{j}$ is an $M$ qubit string.\footnote{$M$ qubit string can be written as the tensor product of $M$ individual qubits. As $j\in \mathbb{Z_N^*}$, $j$ can be expressed in $M=\log N$ bits. Each bit corresponds to a qubit. Thus $\ket{j}$ can be written as $\ket{d_M}^{\otimes M}$, $d_M\in\{0,1\}$ and $M\in[1,\log N]$.}  Alice sends all these $l$ registers to Bob after a random permutation. She keeps the record of the permutation. Bob now operates an oracle $O_t$ on each register. 
The matrix representation of the oracle is as follows.
$$O_t=
\begin{bmatrix}
(-1)^{q_t(0)} & & & \\
 & (-1)^{q_t(1)} & & \\
&&\ddots &\\
 & &  &(-1)^{q_t(N-1)}
\end{bmatrix}
$$
The oracle transforms the $l$-th register $\frac{1}{\sqrt{2}}(\ket{0}+\ket{j})$ to $\frac{1}{\sqrt{2}}(\ket{0}+(-1)^{q_l(j)}\ket{j})$. Bob returns all those registers to Alice. After getting back the registers, Alice measures the decoy registers in $\{\frac{1}{\sqrt{2}}(\ket{0}+\ket{j_i}),\frac{1}{\sqrt{2}}(\ket{0}-\ket{j_i})\}$ basis as she knows the $j_i$ value associates with each register. If any error, which indicates the cheating of Bob, is found, Alice aborts the protocol. Otherwise she will proceed for the next step. 

In the second phase, Alice takes the register which contains $\frac{1}{\sqrt{2}}(\ket{0}+(-1)^{q_t(k)}\\
\ket{k})$ and operates a SWAP gate $U_{swap}$ on the $1$st and $i$-th $1$ of the bit pattern for $k$, $i\in[2,M]$.  She then operates a CNOT gate $U_{cnot}$ on the $1$st and $i$-th $1$. These operations are continued until the bit string for $k$ reduces to $\ket{1}\otimes \ket{0}^{\otimes M-1}$. Thus after these consecutive operations the final state reduces to $\ket{\pm}\ket{0}^{\otimes M-1}$. 

Alice now measures the first particle in $\{\ket{+},\ket{-}\}$ basis. If she gets $\ket{+}$, she concludes that $q_t(k)=0$. If she obtains $\ket{-}$, she concludes that $q_t(k)=1$. Alice sends the value of $t$ and the value of $q_t(k)$ to Bob. Bob checks $p(k)=q_t(k)\oplus r_t$ for that $t$. If $p(k)=1$, Bob concludes that $k$ is a set member of his set. 

In the following section we use this idea to compute set intersection of two computationally indistinguishable sets $X$ and $Y$ holding by Alice and Bob respectively in rational setting.
\section{Proposed Protocol}
In this section we describe the protocol. We assume that Alice and Bob, two rational players, possess two sets $X=\{x_1,x_2,\cdots,x_n\}$ and $Y=\{y_1,y_2,\cdots,y_m\}$ respectively where $x_i, y_i \in \mathbb{Z}_N^*$. The cardinality of $X$ and $Y$ are $n$ and $m$ respectively and are common knowledge to both of the parties. The sets are computationally indistinguishable. 

Now, the players want to compute the intersection of their respective sets. They do not like to reveal any other elements except the intersected ones of their respective sets to the opponent. Each of them has the order of preferences $\mathcal{R}_1 ($\ref{rat}). 

The functionality $\mathcal{F}$ for this game
 can be defined as
 \begin{eqnarray*}
 \mathcal{F}=(X \cap Y, X \cap Y)
\end{eqnarray*}

Our protocol is described in Algorithm~2. The protocol $\bf\Pi$ calls a subroutine $QKeyGen$ to generate a random bit-stream of length $l$, where $l$ is the security parameter. Bob knows the entire bit-stream whereas Alice knows some fraction of this. Our $QKeyGen$ is described in Algorithm~1. The idea of $QKeyGen$ comes from~\cite{Yang}.

For the protocol $\bf\Pi$ we assume that (cooperate, abort) is the  suggested strategy profile for each of the players. That is each player is supposed to follow the protocol and abort if he or she identifies any deviation of his or her respective opponent.

\begin{algorithm}[htbp]
\label{qkeygen}
{\scriptsize
\begin{enumerate}
\item Bob and Alice share $2l$ entangled states of the form
$\frac{1}{\sqrt{2}}(\ket{0}_{B}\ket{\phi_0}_{A}+\ket{1}_{B}\ket{\phi_1}_{A})$, where, $\ket{\phi_{0}}_{A}=\cos{(\frac{\theta}{2})}\ket{0}+\sin{(\frac{\theta}{2})}\ket{1}$ and $\ket{\phi_{1}}_{A}=\cos{(\frac{\theta}{2})\ket{0}}-\sin{(\frac{\theta}{2})}\ket{1}$. Here, subscript B stands for Bob and subscript A stands for Alice. $\theta$ may vary from $0$ to $\frac{\pi}{2}$.
\item Bob measures his qubits in $\{\ket{0}_{B}, \ket{1}_{B}\}$ basis, whereas Alice measures her qubits either in $\{\ket{\phi_{0}}_{A},\ket{\phi_{0}^{\perp}}_{A}\}$ basis or in $\{\ket{\phi_{1}}_{A}, \ket{\phi_{1}^{\perp}}_{A}\}$ basis randomly. 
\item If Bob measures $\ket{0}$, he encodes the bit $r_t$, $t\in [1,2l]$ as $0$. If Bob measures $\ket{1}$, he encodes the bit $r_t$, $t\in [1,2l]$ as $1$. 
\item If the measurement result of Alice gives $\ket{\phi_{0}^{\perp}}$, she concludes that the bit at Bob's end must be $1$. If it would be $\ket{\phi_{1}^{\perp}}$, the bit must be $0$. 
\item Bob and Alice execute classical post-processing in the motivation to check the error in the channel from randomly chosen $l$ bits. If the error remains below the pre-defined threshold, Bob and Alice continues the protocol. Otherwise they abort.  
\item The remaining $l$ bits stream is retained by Bob.  Bob knows the whole stream, whereas Alice generally knows several bits of the stream.   
 
\end{enumerate}
}
\caption{$QKeyGen$}
\end{algorithm}

Before going to the main protocol $\bf\Pi$, we like to explain how Alice gets fraction of the stream. In this direction, we have to calculate the success probability of Alice to guess a single bit possessed by Bob.

As Bob measures his qubits only in $\{\ket{0}_{B}, \ket{1}_{B}\}$ basis, he will get either $\ket{0}$ with probability $\frac{1}{2}$
or $\ket{1}$ with probability $\frac{1}{2}$. When Bob gets $\ket{0}$, Alice should get $\ket{\phi_0}$. If she chooses $\{\ket{\phi_0}_{A}, \ket {\phi_0^{\perp}}_{A}\}$ basis, she will get $\ket{\phi_0}$ with probability $1$ and never gets $\ket{\phi_0^{\perp}}$. However, if she chooses $\{\ket{\phi_1}_{A}, \ket{\phi_1^{\perp}}_{A}\}$ basis, she will get either $\ket{\phi_1}$ with probability $\cos^2\theta$ or $\ket{\phi_1^{\perp}}$ with probability $\sin^2\theta$. We formalize all the conditional probabilities in the following table.

\begin{center}
\begin{tabular}{|c|c|c|c|c|}
\hline
 & \multicolumn{4}{|c|}{Cond. Probability of Alice}\\
\cline{2-5}
& A=$\ket{\phi_0}$ & A=$\ket{\phi_0^{\perp}}$ & A=$\ket{\phi_1}$ & A=$\ket{\phi_1^{\perp}}$\\
\hline
$ B=0 $ & $\frac{1}{2}.1$&$\frac{1}{2}.0$ &$\bf{\frac{1}{2}.\cos^2\theta}$ & $\bf{\frac{1}{2}.\sin^2\theta}$\\
\hline
$B=1$ & $\bf{\frac{1}{2}.\cos^2\theta}$ & $\bf{\frac{1}{2}.\sin^2\theta}$ &$\frac{1}{2}.1$ &$\frac{1}{2}.0$\\
\hline
\end{tabular}
\end{center}

According to the protocol, when Alice gets $\ket{\phi_0^{\perp}}$, she outputs $1$. And when she gets $\ket{\phi_1^{\perp}}$, she outputs $0$. Thus, the success probability of Alice to guess a bit in $l$ bits stream can be written as 
{\scriptsize
$\Pr(A=B)$
\begin{eqnarray*}
\label{sucprob}
&=&\Pr(A=0,B=0)+\Pr(A=1,B=1) \nonumber\\
&=&\Pr(B=0).\Pr(A=0|B=0)+\Pr(B=1).\Pr(A=1|B=1)\\
&=&\frac{1}{2}.\Pr(A=\phi_1^{\perp}|B=0)+\frac{1}{2}.\Pr(A=\phi_0^{\perp}|B=1).
\end{eqnarray*}}
From the above table, we can see that the success probability of Alice becomes $\frac{\sin^2\theta}{2}$. Thus, Alice knows $\frac{\sin^2\theta}{2}$ fraction of the whole stream possessed by Bob.

Now, we describe the protocol $\bf\Pi$ for set intersection in Algorithm~2.

\begin{algorithm}[http]
{\scriptsize
\begin{enumerate}
\item Alice and Bob possess two sets $X=\{x_1,x_2,\cdots,x_n\}$ and $Y=\{y_1,y_2,\cdots,y_m\}$ respectively where $x_i, y_i \in \mathbb{Z}_N^*$; $N\gg2\max(n,m)$~\cite{Shi1}; $u$ is the number of intersected elements. Hence, $u\leq \min(n,m)$
\item Bob now prepares an $N-1$ element database. Any element of the database $p(j)=1$ if and only if $j=y_i, i\in[1,m]$ and $0$ otherwise.
\item Bob calls the sub-routine $QKeyGen$ and prepares a sequence of random bits  $r_1, r_2,\cdots, r_l$, where $l$ is the security parameter and $l\geq 2n$. Alice knows $\frac{\sin^2{\theta}}{2}$ fraction of those bits, where $\theta\in [0,\frac{\pi}{4}]$. 
\item Bob generates a variable $q_t(j)=p(j)\oplus r_t$, where $t$ varies from $1$ to $l$ and $j$ varies from $1$ to $N-1$.
\item Alice and Bob set $p(0)=q_1(0)=q_2(0)=\cdots=q_l(0)=0$ apriori.
\item Alice inserts $n$ check states prepared in $\{\ket{0},\ket{1}\}$ or $\{\frac{1}{\sqrt{2}}(\ket{0}+\ket{1}),\frac{1}{\sqrt{2}}(\ket{0}-\ket{1})\}$ basis randomly. She keeps the record of the positions of the check elements.
\item Alice prepares $2n$ $M(=\log N)$bits registers.
\item For check registers, Alice does the followings
\begin{itemize}
\item If the check element is $\ket{0}$, Alice prepares $\ket{0}^{\otimes M}$
\item If the check element is $\ket{1}$, Alice prepares $\ket{k_s}\ket{0}^{\otimes {M-1}}$; $\ket{k_s}=\ket{1}$, $s\in[1,M]$.
\item If the check element is $\frac{1}{\sqrt{2}}(\ket{0}\pm \ket{1})$, Alice prepares $\frac{1}{\sqrt{2}}(\ket{0}\pm \ket{j})$, where $\ket{0}=\ket{0}^{\otimes M}$ and $\ket{j}=\ket{k_s}\ket{0}^{\otimes{M-1}}$; $k_s=1$ and $s\in\{1,2,\cdots,M\}$.
\end{itemize}
 \item In case of remaining $n$ actual registers, Alice prepares $\frac{1}{\sqrt{2}}(\ket{0}+\ket{j_i})$, $j\in X$ and $i\in[1,n]$. Here, $\ket{0}=\ket{0}^{\otimes M}$ and $\ket{j_i}=\ket{k_1k_2\cdots k_M}$, $k_s\in\{0,1\}$, $s\in\{1,2,\cdots,M\}$. 
\item Alice sends those registers to Bob. If the number of registers exceed $2n$, Bob aborts the protocol.
\item Bob operates the oracle $O_t$~\cite{Shi} on each register and sends those back to Alice. The oracle converts the state $\frac{1}{\sqrt{2}}(\ket{0}+\ket{j_i})$ to $\frac{1}{\sqrt{2}}(\ket{0}+(-1)^{q_t(j_i)}\ket{j_i})$.
\item Alice then selects the check registers and measures those in their respective bases. 
 \item If error is found, she aborts the protocol. Otherwise she will continue.
 \item Alice selects the actual registers, i.e., those registers which contain the actual set elements.
 \item Alice operates $U_{swap}$ followed by $U_{cnot}$ on the first $1$ and $2$nd $1$st $1$ of each register. These operations will be continued till all $M-1$ bits except $k_1$ of each register, reduces to $0$.  In this way, the final state of each register becomes $\ket{\pm}{\ket{0}}^{\otimes {M-1}}$.
 \item Alice measures the first particles of each register in $\{\ket{+},\ket{-}\}$ basis. If she gets $\ket{+}$, she concludes that $q_t(j_i)=0$ and if she gets $\ket{-}$, she concludes that $q_t(j_i)=1$. 
 \item Alice repeats steps $14$, $15$ for all $n$ registers.
\item  Alice conveys the values of $q_t$ and $t$ to Bob for each n
registers.
 \item Bob executes $q_t\oplus r_t$ for that $j$. Note that here Alice knows $j$, but Bob does not. As $r_t$ remains same for each $j$, the value of $q_t(j)$ remains same for each $j$ and flips if $j$ is the set element of $Y$. Thus. without knowing the value of $j$ Bob can conclude if the element is in his set. If he finds $p(j)=0$, he declares that the element is not in his set. If $p(j)=1$, he declares that the element is his set.
 \item Alice checks if the declaration of Bob is correct by finding the value of $p(j)$ for those $t$ for which she has $r_t$. If she finds any cheating of Bob, she aborts the protocol. 
 \item Alice declares the value of the elements for which Bob gets $p(j)=1$.
 \item Bob checks if the element indeed lies in his set. If not, he aborts the protocol.
 \end{enumerate}
}
\caption{Our proposed protocol $\bf{\Pi}$}
\end{algorithm}

\section{Security Analysis} 
\label{sec}
The security criterion of the proposed protocol demands that at the end of the protocol, both Alice and Bob will get $X\cap Y$. However, Alice will not get any element from $Y\setminus (X\cap Y)$ and Bob will not get any element from $X\setminus (X\cap Y)$. Thus, it is quite natural that Alice (resp. Bob) will choose such a strategy which provides her or him the elements from $X\setminus (X\cap Y)$ (resp. $Y\setminus (X\cap Y)$). Hence, without loss of generality, we can discard all other strategies which do not provide such information to Alice (resp. Bob). In the present draft, considering optimised guessing probabilities for Alice and Bob, we show that this security criterion is maintained when the players are rational.

For the security analysis of our protocol we first show if the preferences of Alice and Bob follow the order of $\mathcal{R}_1$, then $((cooperate, abort),(cooperate, abort))$ is a strict Nash equilibrium in the protocol $\bf\Pi$. Then we will show in this initiative both the parties know only $\mathcal{F}$ and nothing else. We also prove that if $((cooperate, abort), (cooperate, abort))$ is a strict Nash, then fairness as well as correctness of the protocol are preserved automatically.

\begin{theorem}
\label{theo1}
In the key establishment phase, $QKeyGen$, of the protocol, for each key bit $r_t$ ($1\leq t\leq 2l$) a dishonest Bob ($\mathcal{B^*}$) can successfully guess if honest Alice ($\mathcal{A}$) gets a conclusive result with probability at most $\frac{1}{2}$.
\end{theorem}

\begin{proof}

Honest Alice ($\mathcal{A}$) follows Algorithm $1$. According to $QKeyGen$, Alice will outputs $r_t=0$, when she chooses $\{\ket{\phi_1},\ket{\phi_1^{\perp}}\}$ basis and gets $\ket{\phi_1^\perp}$. Similarly, Alice outputs $r_t=1$ when she chooses $\{\ket{\phi_0},\ket{\phi_0^{\perp}}\}$ basis and gets $\ket{\phi_0^\perp}$. Each of these two events happens with probability $\frac{1}{2}\sin^2{\theta}$.

Alice and Bob are two distant parties. If we assume no signalling from Alice to Bob, then the basis choice of Alice becomes completely random to Bob. The success probability of honest Bob to guess $r_t$ as a conclusive outcome of Alice is $\frac{1}{2}\sin^2{\theta}$. Dishonest Bob ($\mathcal{B}^*$) always tries to maximize this probability so that he can identify the positions of the bits where Alice gets conclusive results. This information might help him to cheat Alice further. 
Thus, \begin{align*}
\Pr_{guess}[\mathcal{B^*}=\mathcal{A}] \leq \frac{1}{2}
\end{align*}
\end{proof}
\begin{theorem}
\label{min_ent1}
In the key establishment phase of the protocol, for honest Bob ($\mathcal{B}$) and dishonest Alice ($\mathcal{A^*}$), Alice can successfully guess each of the key bits $r_t$ (for all $1 \leq t \leq 2l$) with probability at most $0.85$.
\end{theorem}
\begin{proof} 
At the beginning of key establishment phase dishonest Alice ($\mathcal{A}^*$) and honest Bob ($\mathcal{B}$) share $2l$ copies of entangled pairs. The $t$-th copy of the state is given by $\ket{\psi}_{\mathcal{B}_t\mathcal{A}^*_t }= \frac{1}{\sqrt{2}}(\ket{0}_{\mathcal{B}_t}\ket{\phi_0}_{\mathcal{A}^*_t }+ \ket{1}_{\mathcal{B}_t}\ket{\phi_1}_{\mathcal{A}_t})$, where $t$-th subsystem of Alice and Bob is denoted by $\mathcal{A}^*_t$ and $\mathcal{B}_t$ respectively.

At Alice's side the reduced density matrix is of the form $$\rho_{\mathcal{A}^*_t} =\tr_ {\mathcal{B}_t}[\ket{\psi}_{\mathcal{B}_t\mathcal{A}^*_t}\bra{\psi}] = \frac{1}{2}(\ket{\phi_0}\bra{\phi_0}+\ket{\phi_1}\bra{\phi_1}).$$

At Step $2$, Bob measures each of his part of the state $\ket{\psi}_{\mathcal{B}_t\mathcal{A}_t^*}$ in $\{\ket{0},\ket{1}\}$ basis. Let $\rho_{\mathcal{A}^*_t|r_t}$ denotes the state at Alice's side after Bob's measurement. However, Bob does not communicate his measurement result to Alice. Thus, in this case, for $r_t = 0$, we have $\rho_{\mathcal{A}^*_t|r_t =0} =\frac{1}{2}(\ket{\phi_0}\bra{\phi_0}+\ket{\phi_1}\bra{\phi_1})=\rho_{\mathcal{A}_t^*}$. 
Similarly, for $r_t =1$ we have,
\begin{align*}
\rho_{\mathcal{A}^*_t|r_t =1} & = \tr_{\mathcal{B}_t}[\ket{\psi}_{\mathcal{B}_t\mathcal{A}^*_t}\bra{\psi}]\\
& =  \tr_{\mathcal{B}_t}[\frac{1}{2}(\ket{0}\ket{\phi_0} + \ket{1}\ket{\phi_1})_{\mathcal{B}_t\mathcal{A}^*_t}(\bra{0}\bra{\phi_0}+\bra{1}\bra{\phi_1})]\\
& = \frac{1}{2}(\ket{\phi_0}\bra{\phi_0}+\ket{\phi_1}\bra{\phi_1}\\
&=\rho_{\mathcal{A}_t^*}.
\end{align*}

This implies $\rho_{\mathcal{A}^*_t|r_t} = \rho_{\mathcal{A}^*_t}$. As there is no communication between Alice and Bob, due to non-signalling principle we can claim that Alice can guess Bob's measurement outcome with probability at most $\frac{1}{2}$. This implies, if Alice's optimal guessing strategy is described by the POVM $\{E_z\}_{0\leq z\leq 1}$ then,
\begin{align*}
\Pr_{guess}[r_t|\rho_{\mathcal{A}^*_t}] &= \sum_{r_t}\frac{1}{2}\tr[E_{r_t}\rho_{\mathcal{A}^*_t|r_t}]\\
& = \frac{1}{2}\tr[\sum_{r_t}E_{r_t}\rho_{\mathcal{A}^*_t}]\\
& = \frac{1}{2}.
\end{align*}

However, Alice has the information that if Bob measured $\ket{0}$, i.e., $r_t=0$, her state must collapse to $\ket{\phi_0}$. Similarly, if Bob measured $\ket{1}$, i.e., $r_t=1$, her state collapses to $\ket{\phi_1}$. Thus, if she could distinguish $\ket{\phi_0}$ and $\ket{\phi_1}$ optimally, she can guess $r_t$ optimally. Now, we try to find if this can be done with probability greater than $\frac{1}{2}$.

This distinguishing probability has a nice relationship with the trace distance between the states~\cite{Wilde17}. According to this relation we have,
\begin{align*}
\Pr_{guess}[r_t|\rho_{\mathcal{A}^*_i}] & \leq \frac{1}{2}(1+\frac{1}{2}||\ket{\phi_0}\bra{\phi_0} - \ket{\phi_1}\bra{\phi_1}||_1)\\
& = \frac{1}{2}(1+\sqrt{1 - F(\ket{\phi_0}\bra{\phi_0}, \ket{\phi_1}\bra{\phi_1})})\\
& = \frac{1}{2}(1+\sin{\theta})\\
& = \frac{1}{2} + \frac{1}{2}\sin{\theta}.
\end{align*}

This implies that Alice can successfully guess the value of $r_t$ with probability at most $\frac{1}{2} + \frac{1}{2}\sin{\theta}$.

From the above calculation, it is clear that when $\theta\rightarrow \frac{\pi}{2}$, Alice can get the full information about the stream $r_1,r_2,\cdots, r_l$, and hence can compute $\mathcal{F}$ by herself alone. So, we fix the range of $\theta$ in between $[0,\frac{\pi}{4}]$. For this range of $\theta$, the maximum success probability that Alice can achieve is $0.85$.
\end{proof}
\begin{lemma}(Serfling~\cite{Serfling}) 
\label{serf}
Let $\{x_1,x_2,\cdots,x_n\}$ be a list of values in $[a,b]$ (not necessarily distinct). Let $\overline{x}=\frac{1}{n}\sum_i x_i$ be the average of these random variables. Let $k$ be the number of random variables $X_1,X_2,\cdots,X_k$ chosen from the list without replacement. Then for any value of $\delta>0$, we have
$\Pr\left[|X-\overline{x}| \geq \delta \right] \leq \exp\left(\frac{-2\delta^2 kn}{(n-k+1)(b-a)}\right),$ where $X=\frac{1}{k}\sum_i X_i$.
\end{lemma}
If $k\rightarrow 0$, $\exp\left(\frac{-2\delta^2 kn}{(n-k+1)(b-a)}\right)\rightarrow 1$. For $k=\frac{n}{2}$, we can approximate the probability as
$\Pr\left[|X-\overline{x}| \geq \delta \right] \leq \exp\left(-2\delta^2 n\right).$

Based on the above two theorems and one lemma, we now prove Nash equilibrium of our protocol.
\begin{theorem}
If $(cooperate, abort)$ is the suggested strategy profile for each party and if the parties have the order of preferences $\mathcal{R}_1$, then $((cooperate, abort),(cooperate, abort))$ is a strict Nash equilibrium in the protocol $\bf {\Pi}$ conditioning on $m+n<\frac{N-1}{2}$ and $\theta\in[0,\frac{\pi}{4}]$. 
\end{theorem}
\begin{proof}
Let us consider the deviations of the players from the suggested strategy. It should be noted that when one party deviates, another party follows the protocol.
Now, at first, let Alice deviates from the suggested strategy in the motivation to get $Y\setminus (X\cap Y)$. 
\begin{enumerate}
\item Alice's activities start from step $6$ of algorithm $\bf{\Pi}$. In step $6$, let Alice inserts $n$ extra elements which is not in $X$ but in $\mathbb{Z}_N^*$ along with $n$ check states in the motivation to get $X'\cap Y$, where $X'$ is the set containing $n$ actual elements of set $X$ and $n$ fake elements chosen from $\mathbb{Z}_N^*$. Thus, according to the protocol,  Alice now sends $3n$ registers to Bob.   

Let there exists $X'\cap C_2$, where $C_2=Y\setminus (X\cap Y)$. In this case, Alice will successfully extract those intersected elements in step $19$ of the protocol $\bf{\Pi}$. And hence, security of the protocol will be compromized.

However, Bob knows the cardinality of set $X$. So if Alice tries to send more than $2n$ registers, Bob aborts the protocol and both of them will end up with utility $U^{NN}$. As a result, Alice should have no incentive to choose the above deviation. 

In this regard, one most important thing is that if the cardinalities of $X$ and $Y$ would not be a common knowledge, then choosing the above strategy, Alice might extract some elements from $Y\setminus (X\cap Y)$ causing security loophole in the protocol.

\item Let us now consider a situation when Alice wants to mount the above attack conditioning that the cardinality of her set is a common knowledge. According to the protocol, Bob aborts if he finds more than $2n$ registers coming from Alice. Hence, in this case, Alice will send $2n$ elements from $\mathbb{Z}_N^*$ and no check elements. However, in such situation, Alice can not detect the cheating of Bob.
 
If we assume that Alice sends a few check elements, then also there remains a non-negligible probability for Bob to cheat Alice. This is an immediate instantiation of Serfling Lemma~\ref{serf}. The cardinalities of two sets (one for error checking and another for continuing the protocol) should almost be equal. Hence, Alice has to send at least $n$ check registers to Bob. This resists Alice to send fake elements. If she tries to send fake elements, she has to cut the actual elements. And as a result both will end up with $U^{NN}$.

\item Any type of deviation of Alice in steps $12$, $13$, $14$, $15$, $16$, $17$ and $18$ such as avoiding error checking, operating $U_{swap}$ and $U_{cnot}$ improperly or conveying wrong values of $t$ and $q_{t}$ to Bob will lead her to wrong values of the functionality. So she should have no incentive to deviate in those rounds of the protocol $\bf{\Pi}$.

Maximum what Alice can do in this phase is to send wrong value of $q_t$ for which she has $r_t$. This is because, in such case she can calculate $p(j)$ by herself alone and get the information if the element is in $Y$. By telling wrong value for $q_t$ she can make Bob to calculate wrong value of $p(j)$ and as a result Bob will end up with the functionality $\mathcal{F}$ with less elements. 

However, if $p(j)=0$, Alice should have no incentive to say a wrong $q_t$ as in this case, she knows that the element is not in $Y$. Moreover, if she conveys a wrong $q_t$ for that $j$, then $p(j)$ at Bob's place becomes $1$. And in that case, she has to reveal that element which is in $X$ but not in $Y$. According to the security criterion of the protocol, this situation  is not desirable at all. So, she can only communicate wrong $q_t$ for which she found $p(j)=1$. But by fixing the value of $\theta \in [0,\frac{\pi}{4}]$, we allow Alice to know a few bits of the sequence $r_1r_2\cdots r_l$. Thus, she can only mount such attack for a few elements, but not all.

By choosing this deviation Alice can not even violate the correctness criterion of the protocol. That is Alice can not make Bob to believe in a wrong element as the intersected one. She also can not resist Bob to know most of the elements in $\mathcal{F}$. The above strategy slightly deviates from the suggested strategy for which we get $f=(\mathcal{F}, \mathcal{F})$. Thus this strategy is essentially the same as suggested strategy and does not constitute any deviation.

\item The round, in which deviation may become advantageous to Alice, is step $21$ of protocol $\bf{\Pi}$. 
In step $21$, instead of announcing a correct element of set $X$, she may declare a wrong value for which Bob obtains $p(j)=1$. In this way, she can get the correct values of $\mathcal{F}$ and can deceive Bob to believe in wrong values of $\mathcal{F}$. 

To do this, Alice chooses an element from $\mathbb{Z}_N^*$. But as she does not want to reveal any element except the intersected ones from her set $X$ to Bob, she has to choose an element $e$ from $\mathbb{Z}_N^*$ such that $e\neq x_i$; $x_i\in X$. 

Now, let $S=\mathbb{Z}_N^*\setminus X$, $C_1=\mathbb{Z}_N^*\setminus (X\cup Y)$ and $C_2=Y\setminus (X \cap Y)$. One can write
\begin{eqnarray*}
S=C_1\cup C_2
\end{eqnarray*}

Now, Alice always has to choose an element $e \in S$. If $e$ is a set member of both $S$ and $C_1$, Bob aborts the protocol as he finds that $e\notin Y$. In that case neither Alice nor Bob gets $\mathcal{F}$. The utility functions for Alice and Bob becomes $U_A^{NN}$ and $U_B^{NN}$ respectively. 

However, if $e$ belongs to $S$ and $C_2$, Bob can not distinguish if $e\in X\cap Y$ or $e\in Y\setminus (X \cap Y)$. So he does not abort the protocol as $e \in Y$. In this case, Alice knows the correct elements but Bob ends up with wrong elements which is effectively equivalent to obtaining no element as this element $e$ neither belongs to $X \cap Y$ nor is in $X\setminus (X \cap Y)$. Thus, in this case, the utility of Alice becomes $U_A^{TN}$.

Moreover, correctness of the protocol is violated, i.e.,
\begin{eqnarray*}
\Pr[f_B\not\in \{\mathcal{F},\perp\}|A=\sigma'_A, B=\sigma_B]\neq 0.
\end{eqnarray*}

 However, Alice does not know $Y$. Let $u$ be the number of intersected elements. Thus, $|C_1|=N-1-n-m+u$ and $|C_2|=m-u$. Probability that $e$ is in set $C_1$ is
 \begin{eqnarray*}
 \Pr(e\in C_1)&=&\frac{N-1-n-m+u}{N-1-n}\\
 &=& 1-\frac{m-u}{N-1-n}
 \end{eqnarray*}
 Probability that $e$ is in set $C_2$ is
 \begin{eqnarray*}
 \Pr(e\in C_2)&=&\frac{m-u}{N-1-n}\\
 \end{eqnarray*}
 Now, 
 \begin{eqnarray*}
 n+m&\leq& 2\max(n,m)\\
\Rightarrow n+m-u&\leq& 2\max(n,m); u\leq \min(n,m)\\
 \Rightarrow n+m-u&\ll&N-1\\
 \Rightarrow m-u&\ll&N-1-n
 \end{eqnarray*}
 Let $\epsilon=\frac{m-u}{N-1-n}$. Then, with probability $1-\epsilon$, in step $21$ of protocol $\bf{\Pi}$, for some of $p(j)$s, Alice will choose $e$ from the set $C_1$. In that case, Bob immediately aborts the protocol. Hence, the protocol will be terminated and Alice will not get any intersected elements further \footnote{We assume that there is no payoff for a player who deviates from the game and gets partial knowledge about the functionality. In this case, partial knowledge is considered as no knowledge or $\perp$.}. So, the expected utility $E(U_A)$ over this deviation can be expressed as
 \begin{eqnarray*}
 E(U_A)=\Pr(e\in C_2) U_A^{TN} +  \Pr(e\in C_1) U_A^{NN}
\end{eqnarray*}
Depending on the values of $U_A^{TN}$, $U_A^{TT}$ and $U_A^{NN}$ we can fix the values of $m$, $n$ and $N$ in such a way so that
\begin{eqnarray}
\label{eqn}
\Pr(e\in C_2) U_A^{TN} +  \Pr(e\in C_1) U_A^{NN}< U_A^{TT}
\end{eqnarray}
For example, let $U_A^{TN}=1$, $U_A^{NN}=0$ and $U_A^{TT}=\frac{1}{2}$. Then equation~\ref{eqn} reduces to
  \begin{eqnarray*}
\Pr(e\in C_2) < \frac{1}{2}\\
\Rightarrow \frac{m-u}{N-1-n}<\frac{1}{2}\\
\Rightarrow 2(m-u)< N-1-n
\end{eqnarray*}
Putting $u=0$, we get
\begin{eqnarray*}
2m<N-1-n.
\end{eqnarray*} 
Thus, if $m+n<\frac{N-1}{2}$, then for any value of $u=\min(n,m)$, 
 we can bound 
Alice to choose cooperation over such deviation. Hence, we can write
{\small
\begin{eqnarray*}
&U_A&((cooperate,abort), (cooperate,abort))\\
>&U_A&(deviation, (cooperate,abort))
\end{eqnarray*}}

\end{enumerate}
Now, we focus on Bob's deviations. Bob will deviate in the motivation to extract the elements from $X\setminus (X\cap Y)$. The possible deviation in this case is to tell a wrong value of $p(j)$ at step $19$, i.e., when he gets $p(j)=0$, he declares $p(j)=1$. Bob can declare $p(j)=1$ for all $n$ registers. In this case, he actually comes to know the values of all the set elements of $X$. Hence, the security criterion that Bob should not be allowed to know the set elements of Alice other than the intersected ones, is violated.

However, Alice possesses some bits of the stream $r_1r_2\cdots r_l$. Alice gets the values of $q_t$ for all $n$ registers. For the cases where she knows $r_t$, she can easily calculate the value of $p(j)$ and can check if those values match with the values declared by Bob. If not, Alice aborts the protocol without announcing any element further (step $20$). 

On the other hand, Bob has at most a random guess about $t$, i.e., the position of Alice's conclusive result. So, if he tries to declare wrong values for $p(j)$, with probability $\frac{1}{2}$, he will be caught by Alice and protocol will be terminated.

Again, Alice does not know the whole bit stream. So, she can not calculate $p(j)$ for all $n$ registers by herself alone and will not be able to get $\mathcal{F}$ completely by knowing $q_
t$ only. In this case, she can get a very few elements which is equivalent to obtaining $\perp$. Hence, both Alice and Bob will end up with utility $U^{NN}$. According to $\mathcal{R}_1$, $U_B^{TT}> U_B^{NN}$, Bob has no incentive to follow the deviation. Rather he prefers cooperation. Thus, for Bob also we can write
{\small
\begin{eqnarray*}
&U_B&((cooperate,abort), (cooperate,abort))\\
>&U_B&((cooperate,abort), deviation)
\end{eqnarray*}}
This completes the proof.
\end{proof}
\begin{theorem}
If $((cooperate, abort), (cooperate, abort))$ is a strict Nash, then fairness of the protocol is guaranteed.
\end{theorem}
\begin{proof}
The strategy vector ((cooperate, abort), (cooperate, abort)) is strict Nash implies that 
\begin{eqnarray*}
\Pr[f_A=\mathcal{F}|A=\sigma'_A, B=\sigma_B] \\
< \Pr[f_A=\mathcal{F}|A=\sigma_A, B=\sigma_B].
\end{eqnarray*}

where $\sigma'_A$ denotes any deviation by Alice. Similarly, for Bob we can write
\begin{eqnarray*}
\Pr[f_B=\mathcal{F}|A=\sigma_A, B=\sigma'_B] \\
< \Pr[f_B=\mathcal{F}|A=\sigma_A, B=\sigma_B].
\end{eqnarray*}
where $\sigma'_B$ denotes any deviation by Bob.
\end{proof}
\begin{theorem}
If $((cooperate, abort), (cooperate, abort))$ is a strict Nash, then correctness of the protocol is guaranteed.
\end{theorem}
\begin{proof}
The proof is immediate. As ((cooperate, abort), (cooperate, abort)) is a strict Nash, 
\begin{eqnarray*}
\Pr[f_A\not\in \{\mathcal{F},\perp\}]=\Pr[f_B\not\in \{\mathcal{F},\perp\}]= 0.
\end{eqnarray*}
\end{proof}
 Our next job is to prove the security when Alice (resp. Bob) can mount an active attack. By the word ``active attack'', we want to mean that though Alice and Bob follow the suggested strategies,  but exploiting the advantage of quantum theory, they may steal some information which is restricted by the protocol. We now show that Alice as well as Bob know only $\mathcal{F}$ and nothing else in this rational setting. The analysis goes on the same line of~\cite{Shi}.
 
 The goal of Alice and Bob is to know the elements, other than the intersected ones,
 of the sets of their respective opponents. In Algorithm $1$, Bob sends one part of each entangled pairs to Alice. So Alice can mount an active attack in this phase. By considering optimal POVM, she can increase her success probability to guess a bit of the bit-stream $r_1r_2\cdots r_l$. This attack has been considered in theorem~\ref{min_ent1}.
 
 In protocol $\bf{\Pi}$, Bob does not send any elements to Alice. However, Alice sends all the set elements in an encrypted form to Bob. So, in this phase, it is Bob who can perform an active attack on Alice's system.
 
 One most simple and common attack is measure and resend attack. In this attack model Bob measures each register in some bases and then prepares that register in the measured state. In this way he tries to extract some information about the set elements of Alice. 
 
 To detect such type of attack we set $n$ check registers. Those are prepared either in $\{\ket{0},\ket{1}\}$ or in $\{\frac{1}{\sqrt{2}}(\ket{0}+\ket{1}),\frac{1}{\sqrt{2}}(\ket{0}-\ket{1})\}$ basis. As Bob does not know the positions of those check registers, he can not bypass those and performs the attack only on the actual registers containing the set elements. So, checking the error rate for the check registers, Alice can identify the attack.
 
 Due to no cloning~\cite{wootters,dieks} theorem and Heisenberg uncertainty principle, Bob can not distinguish the check registers with probability one. If he tries to distinguish them, he must incorporate some noise in the system. 
 
On the other hand, Alice knows the bases for those states. So, she can measure the returning registers in perfect bases, i.e., either in $\{\ket{0},\ket{1}\}$ or in $\{\frac{1}{\sqrt{2}}(\ket{0}+\ket{1}), \frac{1}{\sqrt{2}}(\ket{0}-\ket{1})\}$ basis. Now, the oracle $O_t$, $t\in[1,l]$, converts 
\begin{itemize}
\item $\ket{0}$ to $\ket{0}$
\item $\ket{1}$ which is encoded as $\ket{k_s}\ket{0}^{\otimes {M-1}}$; $k_s=1$ and $s\in[1,M]$ , to $(-1)^{q_t(k_s)}\ket{k_s}\\ \ket{0}^{\otimes {M-1}}$
\item $\frac{1}{\sqrt{2}}(\ket{0}\pm \ket{1})$ which is encoded as $\frac{1}{\sqrt{2}}(\ket{0}\pm \ket{j})$, where $\ket{0}=\ket{0}^{\otimes M}$ and $\ket{j}=\ket{k_s}\ket{0}^{\otimes{M-1}}$; $k_s=1$ and $s\in\{1,2,\cdots,M\}$, to $\frac{1}{\sqrt{2}}(\ket{0} \pm (-1)^{q_t(k_s)}\ket{j})$ 
\end{itemize}
Hence, if Alice sends $\ket{0}$, in case of no attack, she should get $\ket{0}$. If she gets $\ket{1}$, she concludes that the attack has been mounted. Same thing happens for $\ket{1}$. However, the attack can not be distinguished for the states $\frac{1}{\sqrt{2}}(\ket{0}+\ket{1})$ and $\frac{1}{\sqrt{2}}(\ket{0}-\ket{1})$. This is because, without any attack the state $\frac{1}{\sqrt{2}}(\ket{0}+\ket{1})$ may convert to $\frac{1}{\sqrt{2}}(\ket{0}-\ket{1})$ and vice versa due to the effect of the oracle $O_t$. But Alice can always detect the noise in $\{\ket{0},\ket{1}\}$ basis. If it is above the threshold, Alice aborts the protocol. 
 
 In this regard, we like to emphasize that in~\cite{Shi} it is commented that this attack can be identified by measuring the returned decoy states in $\{\frac{1}{\sqrt{2}}(\ket{0}+\ket{j_i}),\frac{1}{\sqrt{2}}(\ket{0}-\ket{j_i})\}$ bases, $i\in [1,l]$, $j_i\in Z_N^*$. Bob will measures the states in computational basis. If he gets $\ket{j_i}$, he can prepare the state as $\frac{1}{\sqrt{2}}(\ket{0}+\ket{j_i})$ or $\frac{1}{\sqrt{2}}(\ket{0}-\ket{j_i})$. As a result Alice can not detect if the attack performed in the system. However, if he gets $\ket{0}$, he can not create perfect superposition. And Alice can detect the attack measuring the states in $\{\frac{1}{\sqrt{2}}(\ket{0}+\ket{j_i}),\frac{1}{\sqrt{2}}(\ket{0}-\ket{j_i})\}$ basis. However, the oracle has been designed in such a way so that it can map $\frac{1}{\sqrt{2}}(\ket{0}+\ket{j_i})$ into $\frac{1}{\sqrt{2}}(\ket{0}-\ket{j_i})$. Alice does not know $q_t(j)$. So it is not possible for her to determine when she would get $\frac{1}{\sqrt{2}}(\ket{0}+\ket{j_i})$ or $\frac{1}{\sqrt{2}}(\ket{0}-\ket{j_i})$ apriori. In that case, it is not very clear how Alice can distinguish if the attack has been performed or the oracle has been operated on the states. To avoid such ambiguity, we prepare the check elements in $\{\ket{0},\ket{1}\}$ and $\{\frac{1}{\sqrt{2}}(\ket{0}+\ket{1}),\frac{1}{\sqrt{2}}(\ket{0}-\ket{1})\}$ bases randomly so that Alice can distinguish whether attack has been mounted or oracle has been operated checking the noise in $\{\ket{0},\ket{1}\}$ basis. 
 
 In~\cite{Shi}, the authors analyze a more complicated attack known as entanglement measure attack. In this attack model Bob combines an ancillary state with each register. He then performs a unitary operation on the register and the ancillary state. The unitary operation $Q$ is described as follows.
 \begin{eqnarray*} 
 Q_{AB}\ket{0}_A\ket{0}_B&=&\sqrt{\eta}\ket{0}_A\ket{\phi_0}_B+\sqrt{1-\eta}\ket{V_0}_{AB}\\
 Q_{AB}\ket{v}_A\ket{0}_B&=&\sqrt{\eta}\ket{v}_A\ket{\phi_v}_B+\sqrt{1-\eta}\ket{V_v}_{AB}\\
 \end{eqnarray*}
 Thus, one can write
 \begin{eqnarray*}
 Q_{AB}\frac{1}{\sqrt{2}}(\ket{0}_A+\ket{v}_A)\ket{0}_B&=&\sqrt{\eta}(\ket{0}_A\ket{\phi_0}_B
 +\ket{v}_A\ket{\phi_v}_B)\\
 &+&\sqrt{1-\eta}(\ket{V_0}_{AB}+\ket{V_v}_{AB}).
 \end{eqnarray*}
 Where $A$ stands for Alice's system and $B$ stands for Bob's system; $\eta$ is some probability. Here,
 \begin{eqnarray*}
 \langle{0}{\phi_0}|{V_0}\rangle_{AB}=\langle{v}{\phi_v}|{V_v}\rangle_{AB}=\langle{0}{\phi_0}|{V_v}\rangle_{AB}= \langle{v}{\phi_v}|{V_0}\rangle_{AB}=0.\\
 \end{eqnarray*}
  After applying the oracle he then sends the registers back to Alice and keeps the ancillary systems with him. He measures the ancillary systems to extract the information about the states of the registers. In this initiative, Shi et al.~\cite{Shi} bound the amount of information extracted by Bob by fixing the threshold value sufficiently small.
  
In our case, we redefine this attack models as follows.
\begin{eqnarray*}
Q_{AB}(\ket{\psi},\ket{W}) = \sqrt{\eta}\ket{\psi}\ket{E_{00}} + \sqrt{1-\eta}\ket{\psi}^{\perp}\ket{E_{01}}\\
Q_{AB}(\ket{\psi}^{\perp},\ket{W}) = \sqrt{1-\eta}\ket{\psi}\ket{E_{10}} + \sqrt{\eta}\ket{\psi}^{\perp}\ket{E_{11}}
\end{eqnarray*}
where, $\ket{\psi}$ is any arbitrary qubit of the form $\cos{\frac{\theta}{2}}\ket{0}+\sin{\frac{\theta}{2}}\ket{1}$, $\theta\in [0,\frac{\pi}{2}]$ and $\ket{\psi}^{\perp}$ is the orthogonal state of $\ket{\psi}$. $W$ is the ancillary state inserted by Bob. $E_{u,v}$, $u,v\in\{0,1\}$, are the states possessed by Bob after the application of $Q_{AB}$. Here, we assume that 
\begin{eqnarray*}
\langle E_{00}|E_{01}\rangle=\langle E_{10}|E_{11}\rangle=0.
\end{eqnarray*}
For the check elements in $\{\ket{0},\ket{1}\}$ basis, the above equations reduce to 
\begin{eqnarray*}
Q_{AB}(\ket{0},\ket{W}) = \sqrt{\eta}\ket{0}\ket{E_{00}} + \sqrt{1-\eta}\ket{1}\ket{E_{01}}\\
Q_{AB}(\ket{1},\ket{W}) = \sqrt{1-\eta}\ket{0}\ket{E_{10}} + \sqrt{\eta}\ket{1}\ket{E_{11}}
\end{eqnarray*}
When Bob returns those registers to Alice, Alice measures those in $\{\ket{0},\ket{1}\}$ basis. She knows when $\ket{0}$ (resp. $\ket{1}$) has been sent. If she gets the orthogonal states of the states sent to Bob, she concludes that the attack has been mounted and aborts the protocol.
\section{Communication Complexity of the Protocol}
In this section, we compute the communication complexity of the proposed protocol. Like most of the quantum protocol, in our protocol also we have to communicate qubits as well as classical bits. So, we can divide the communication complexity into two parts; one is quantum communication complexity and another is classical communication complexity.  

 In $QKeyGen$ part the total communication complexity is $2l$, as $2l$ entangled qubits are sent from Bob to Alice. The total quantum communications required in the protocol $\bf\Pi$ is $4n$; $2n$ quantum registers are sent from Alice to Bob and $2n$ quantum registers are returned back from Bob to Alice. Thus, the overall quantum communication complexity is $(4n+2l)$. 

 After error estimation phase, Alice finds $q_t$ and $t$ for each $n$ registers. The value of $t$ is expressed in $\log l$ bits and $q_t\in\{0,1\}$. So the total communicated bits in step $17$ is $n (\log l+1)$. In step $18$, Bob declares the value of $p(j)\in\{0,1\}$ for each $n$ registers. Thus, in step $18$, there are $n$ classical communications. In step $19$, Alice announces the value of the set elements for which Bob declared $p(j)=1$. We assume that there are $u$ intersected elements. So, if we express the value of each $x_i\in X$ in $\log N$ bits, then the total number of communicated bits should be $u\log N$. Hence, the overall classical communication complexity becomes $(n(\log l+2)+u\log N)$.
 
 In this regard, we like to compare the communication complexity with other similar or classical schemes, to put the protocol in perspective.
 
 In classical domain Freedman, Nissim and Pinkas~\cite{freedman} studied set intersection problem in semi-honest setting. The sets in their protocol include $0$. The communication complexity of the protocol is $O(m_X+m_Y)$, where $m_X$ and $m_Y$ are the cardinality of the sets considered.
 
 Hazay and Lindell~\cite{Hazay} revisited the set intersection problem in the motivation to propose an efficient protocol against a more realistic adversary than semi-honest adversary. In this direction, they proposed two protocols; one against a malicious adversary and other against a covert adversary. Both the protocols are constant round and incur the communication of $O(m_X.p(n)+m_Y)$ group elements, where $p(n)$ is polynomial in the security parameter $n$.
 
 The protocol proposed by Dachman-Soled, Malkin, Raykova and Yung~\cite{dachman} for set intersection in the presence of malicious adversaries incurs communication of $O(m_Y n^2 log^2 m_X + m_X n)$ group elements.
 
 In quantum domain, the communication complexity of {\em Quantum Oblivious Set Member Decision Problem} by Shi et al. ~\cite{Shi}, is constant, i.e., $O(1)$. However, one should note that none of these schemes considered rational adversaries. 

\section{Conclusion} In the present draft we propose a two party protocol for computing set intersection securely in quantum domain. The parties, Alice and Bob, have two sets $X$ and $Y$ which are computationally indistinguishable from each other. In classical domain this problem has been considered in~\cite{HZ1,HZ2,JN10}. However, the hardness assumptions exploited in those works are proven to be vulnerable in quantum domain. 

In quantum domain Shi et al.~\cite{Shi} proposed a variant of this problem and named it as {\em Quantum Oblivious Set Member Decision Protocol} (QOSMDP). We extend this problem to compute set intersection of two computationally indistinguishable sets. We consider rational setting as rational setting is more realistic  than being completely honest or completely malicious. 

In rational setting, we prove that if $(cooperate, abort)$ is the suggested strategy profile for each of the two players, then $((cooperate, abort), (cooperate, abort))$ is a strict Nash equilibrium in our protocol. Following the lines of security proof of~\cite{Shi}, we also show that in this initiative Alice and Bob only know $\mathcal{F}$ and nothing else, i.e., Alice does not know any element of the set $Y\setminus(X \cap Y)$ and Bob does not know any element from the set $X\setminus(X \cap Y)$. We also prove that if $((cooperate, abort), (cooperate, abort))$ is a strict Nash, then fairness and correctness of the protocol are preserved.

\end{document}